\documentclass[conference]{IEEEtran}
\IEEEoverridecommandlockouts
\usepackage{cite}
\usepackage{amsmath,amssymb,amsfonts}
\usepackage{algorithmic}
\usepackage{graphicx}
\usepackage{textcomp}
\usepackage{xcolor}
\usepackage{subfigure}
\usepackage{physics}
\usepackage{pdfpages}
\usepackage{amsmath,amssymb,amsfonts,bm,amsthm}
\usepackage{amsthm}
\newtheorem{theorem}{Theorem}
\newtheorem{lemma}{Lemma}

\theoremstyle{definition}
\newtheorem{definition}{Definition}
\usepackage{svg}
\usepackage{float}
\usepackage[numbers,sort&compress]{natbib}

\def\BibTeX{{\rm B\kern-.05em{\sc i\kern-.025em b}\kern-.08em
    T\kern-.1667em\lower.7ex\hbox{E}\kern-.125emX}}
\begin{document}

\title{
Harnessing Inherent Noises for Privacy Preservation in Quantum Machine Learning
%Quantum Inherent Noise for Privacy Preservation in Quantum Machine Learning
}

\author{
\IEEEauthorblockN{Keyi Ju$^{1}$, Xiaoqi Qin$^{1}$, Hui Zhong$^{2}$, Xinyue Zhang$^{3}$, Miao Pan$^{2}$, Baoling Liu$^{1}$}
\IEEEauthorblockA{$^{1}$State Key Laboratory of Networking and Switching Technology,\\ Beijing University of Posts and Telecommunications, Beijing, China, 100876}
\IEEEauthorblockA{$^{2}$Department of Electrical and Computer Engineering, University of Houston, Houston, TX, 77204}
\IEEEauthorblockA{$^{3}$Department of Computer Science, Kennesaw State University, Marietta, GA, 30060}
}

\maketitle

\begin{abstract}
Quantum computing revolutionizes the way of solving complex problems and handling vast datasets, which shows great potential to accelerate the machine learning process. However, data leakage in quantum machine learning (QML) may present privacy risks. Although differential privacy (DP), which protects privacy through the injection of artificial noise, is a well-established approach, its application in the QML domain remains under-explored. In this paper, we propose to harness inherent quantum noises to protect data privacy in QML. Especially, considering the Noisy Intermediate-Scale Quantum (NISQ) devices, we leverage the unavoidable shot noise and incoherent noise in quantum computing to preserve the privacy of QML models for binary classification. We mathematically analyze that the gradient of quantum circuit parameters in QML satisfies a Gaussian distribution, and derive the upper and lower bounds on its variance, which can potentially provide the DP guarantee. Through simulations, we show that a target privacy protection level can be achieved by running the quantum circuit a different number of times.
%Quantum computing has revolutionized the approach to solving complex problems and handling vast datasets, potentially accelerating the machine learning process. However, data leakage in quantum machine learning (QML) may present privacy risks. While differential privacy (DP) has been a classical solution to protect privacy by injecting artificial noise, the implementation of DP within the QML domain remains an under-explored area. In this paper, we analyze the role of inherent quantum noise in protecting model parameters, considering the Noisy Intermediate-Scale Quantum (NISQ) era, by combining the unavoidable shot noise and incoherent noise in quantum computing with binary classification machine learning task. We mathematically analyze that under certain conditions, the gradient of the quantum circuit parameters satisfies a Gaussian distribution and give upper and lower bounds on the standard deviation. We also simulate the privacy budget that can be provided with different parameter settings.
\end{abstract}


\section{Introduction}
Quantum computing, which exploits principles of quantum mechanics such as superposition and entanglement, offers a great potential for significantly faster complex calculations than classical computers.
% Quantum computing is a groundbreaking paradigm that utilizes the great features of quantum mechanics to process information, promising to be revolutionary computational tool.
% ~\cite{steane1998quantum}.
%Quantum computing represents a groundbreaking paradigm that harnesses the principles of quantum mechanics to process information, promising to revolutionize computational power~\cite{steane1998quantum}.
% Within this realm, variational quantum algorithms (VQAs) have emerged as a powerful approach, integrating classical optimization techniques with quantum circuits, thereby facilitating the solution of complex computational tasks that are currently intractable for classical computers~\cite{cerezo2021variational}. 
In the Noisy Intermediate-Scale Quantum (NISQ) era, variational quantum algorithms (VQAs) are considered the best for quantum computing, 
% because VQAs are noise tolerant compared to other algorithms and give quantum superiority with only a few hundred qubits.
which heavily rely on variational quantum circuits (VQCs). VQCs are parameterized quantum circuits designed to be optimized iteratively to minimize specific cost functions, enabling the exploration of solution spaces and the identification of optimal quantum states. Based on the advance of these technologies, quantum computing is becoming a powerful alternative to conduct computing intensive machine learning or optimization tasks~\cite{ramezani2020machine,ajagekar2019quantum}.

%many significant areas such as machine learning \cite{ramezani2020machine}, optimization \cite{ajagekar2019quantum}, and healthcare \cite{ur2023quantum}. 

%Among all these respects, 
Quantum machine learning (QML) has demonstrated advantages over classical machine learning (CML) in dealing with large datasets and complex problems \cite{biamonte2017quantum}. 
% This efficiency enhancement is attributed to quantum transformations combined with QML, which improve model training and inference compared to CML. 
Due to its high computing efficiency, QML has the potential to be used in various areas such as nanoparticle synthesis and
% ~\cite{tao2021nanoparticle}
biomedical domain.
% ~\cite{maheshwari2022quantum}
Nevertheless, the datasets used in these applications may be highly sensitive, e.g., CT scanned images of COVID-19 patients~\cite{sengupta2021quantum}. Therefore, similar to data privacy concerns in CML, privacy preservation is crucial in QML to safeguard sensitive data from potential breaches and misuse.

%ensuring the protection of individual rights and security, just like CML.

In CML scenarios, we integrate differential privacy (DP) techniques to protect sensitive data and mitigate privacy risks during data processing and analysis. DP is a concept in data science and statistics that aims to provide a means of statistical database privacy protection proposed by~\cite{dwork2014algorithmic}. It offers a way to maximize the accuracy of queries from statistical databases while minimizing the chances of identifying information about specific individuals within the database. By introducing controlled amounts of statistical noise to the data, DP ensures that the presence or absence of specific records does not significantly affect the results, thereby safeguarding individual privacy. One of the most popular DP definitions is ($\epsilon,\delta$)-DP, where $\epsilon$ represents the maximum allowed change in output due to the addition or removal of an individual's data, and $\delta$ indicates an upper bound on the probability of the algorithm's privacy violation. 
%which is quantified by $\epsilon$ and $\delta$, with $\epsilon$ representing the maximum allowed change in output due to the addition or removal of an individual's data, and $\delta$ providing an upper bound on the probability of the algorithm's privacy violation. 
Abadi \textit{et al.} in~\cite{abadi2016deep} proposed a differentially private stochastic gradient descent (DP-SGD) algorithm, which can provide a strict privacy guarantee in machine learning by adding Gaussian noise to gradients.
%. It gain a lot of popularity and remains the standard method implemented in many DP-Training libraries.
% Differential privacy (DP) is a concept in data science and statistics that aims to provide a means of statistical database privacy protection which is proposed in~\cite{dwork2014algorithmic}. 
%DP has proven to be very 
% such as adversarial examples,
% ~\cite{lecuyer2019certified}
%membership inference,
% ~\cite{song2019auditing}
%model stealing,
% ~\cite{krishna2019thieves}
%and model inversion attacks.
% ~\cite{zhang2020secret}
% adversarial examples~\cite{lecuyer2019certified}, membership inference~\cite{song2019auditing}, model stealing~\cite{krishna2019thieves}, and model inversion~\cite{zhang2020secret}.

Building on this foundation, research efforts have been dedicated to generalizing DP for quantum computing scenarios. Aiming to preserve the privacy of quantum data and computations, 
% some pioneering works~\cite{hirche2023quantum,guan2023detecting,watkins2023quantum,zhou2017differential} tried to extend the concept of DP in CML to the realm of quantum computing or QML.
%In recent years, quantum differential privacy (QDP), first proposed by \cite{zhou2017differential}, brings DP to the quantum world. QDP extends the concept of DP to the realm of quantum computing, aiming to preserve the privacy of quantum data and computations~\cite{hirche2023quantum,guan2023detecting,watkins2023quantum}.
% For example, Zhou and Ying in~\cite{zhou2017differential} considered the effect of amplitude damping, phase damping and depolarizing noise on the output of quantum circuits, which in turn proves that the output satisfies a certain protection level of DP, but there is no guarantee that the parameters of VQCs are protected by DP and they didn't extend QDP to QML. 
%Watkins \textit{et al.} in~\cite{watkins2023quantum} proposed to protect the privacy of QML models by artificially adding Gaussian noise to the model parameters utilizing the DP-SGD algorithm and protecting the privacy of QML models.
Watkins \textit{et al.} in~\cite{watkins2023quantum} proposed a method to protect the data privacy of QML models by adding artificial Gaussian noise to the gradients using the DP-SGD algorithm.
%However, it largely ignores the existence of inherent noise of quantum circuits, which can potentially provide the privacy protection guarantee in QML.
However, this method largely ignores the inherent noises in quantum circuits, which could potentially offer a privacy protection guarantee in QML. Apart from them, we find that DP can also be achieved with shot noise that are inherent in quantum measurements. The shot noise are statistical fluctuations in the measurement outcomes of quantum bits (qubits) due to the discrete nature of particles.
%which refer to the statistical fluctuations in the measurement outcomes of quantum bits (qubits) that arise due to the discrete nature of particles. 
Fluctuations caused by shot noise are intrinsic to the quantum nature of particles and result from the probabilistic behavior of quantum systems. Shot noise affects the output of quantum circuits and makes it Gaussian distributed, potentially preserving the privacy of QML models.

In this paper, we theoretically study the impacts of shot noise on the privacy budget of $(\epsilon,\delta)$-DP. Meanwhile, considering that QML is severely affected by incoherent noise, we propose to employ the quantum error mitigation (QEM) method 
%which is commonly deployed in contemporary quantum computing 
to mitigate the incoherent noise, and investigate the impacts of the QEM method on the degree of privacy preservation. The methods of QEM include zero noise extrapolation, probabilistic error cancellation (PEC) \cite{cai2022quantum}, etc. In this study, we utilize PEC as our QEM method, since the analytic representation helps us analyze the impact of QEM method on the quantum circuit outputs. To the best of our knowledge, we are the first to utilize the inherent noises of quantum circuits to achieve DP preservation in QML.
% Quantum computing is currently in the Noisy Intermediate-Scale Quantum (NISQ) era, characterized by the existence of intermediate-scale quantum devices that are subject to errors and noise during computation. Quantum Error Mitigation (QEM) is a critical area of research in the NISQ era, aimed at addressing the inherent errors and noise present in quantum computing systems. The methods of QEM include zero noise extrapolation \cite{giurgica2020digital}, measurement error mitigation \cite{barron2020measurement}, probabilistic error elimination (PEC) \cite{temme2017error}, etc. Despite the noise limitations, quantum devices in the NISQ era are still capable of performing specific tasks that surpass the capabilities of classical computers.
% Such inherent noise due to quantum computing can potentially be used as a source of noise for realizing QDP, turning errors into treasures. Since it is difficult to quantify the various types of noise in quantum computing, we designate the intrinsic noise as depolarizing noise which can be a good approximation of all types of errors in sufficiently deep quantum circuits \cite{vovrosh2021simple}. 
The contributions of this paper are three-fold as summarized below:
\begin{itemize}
    \item With shot noise, incoherent noise, and the PEC method, we show that the outputs of a quantum circuit satisfy Gaussian distribution, and provide its analytical expressions.
    %We illustrate that under the effect of shot noise, global depolarizing noise and the PEC method, the outputs of a quantum circuit satisfy Gaussian distribution and we give its analytical expressions.
    \item We further prove that the gradient of QML also satisfies Gaussian distribution, when performing a binary classification task with a modified hinge loss function and considering the inherent noise. We also give upper and lower bounds on the variance of the Gaussian distribution, which is closely related to the privacy budget $\epsilon$ and the privacy violation probability $\delta$.    
    %It is shown that when a binary classification task is performed using QML and the loss function is hinge loss, the gradient of QML also satisfies Gaussian distribution. We give upper and lower bounds on the variance of the Gaussian distribution, which is closely related to the privacy budget $\epsilon$ and the privacy violation probability $\delta$.
    \item We conduct extensive simulations on the quantum simulator, and the results demonstrate the inherent noises in quantum circuits can be harnessed to implement DP. Moreover, we show that under a certain $\delta$, the privacy budget $\epsilon$ can be achieved by running the quantum circuit a different number of times.
    %We demonstrate through simulation experiments that the noise inherent in quantum circuits can realize DP and under a certain $\delta$, the privacy budget $\epsilon$ can be controlled by controlling shots of the quantum circuits.
\end{itemize}

The rest of the paper is organized as follows. Section \ref{sec:dp} introduces the background knowledge of DP and DP-SGD algorithm. We give analyze the effect of shot noise, incoherent noise, and the PEC method on the gradient of QML and further discuss the relationships between the distribution of the gradient and the privacy budget of DP in Section \ref{Sec.3}. Then, we present the experiment study in Section \ref{EVALUATION}. Finally, we conclude our paper in Section \ref{Sec.5}.

\section{Differential Privacy and DP-SGD Algorithm}\label{sec:dp}
In this section, we introduce the basic concept of DP and the DP-SGD algorithm proposed by Abadi \textit{et al.} in~\cite{abadi2016deep} achieving DP of machine learning models.
\subsection{Differential Privacy}
DP~\cite{dwork2008dp} is a method for preserving the privacy of a sensitive dataset while maintaining statistical information of the dataset. Intuitively, given a private dataset, a DP mechanism can protect individual data privacy even though statistics of the dataset are published. The formal definition of $(\epsilon,\delta)$-DP is shown as follows.
\begin{definition}[\bf{Differential Privacy}]
A randomized algorithm $\mathcal{M}$ satisfies $(\epsilon,\delta)$-DP if for any two adjacent datasets $\mathcal{D}$ and $\mathcal{D}^\prime$ that differ in only a single record, and for all $S \subseteq Range(\mathcal{M})$, we have 
  \begin{equation}
        \text{P}[\mathcal{M}(\mathcal{D})\in S] \leq e^{\epsilon}\text{P}[\mathcal{M}(\mathcal{D}^\prime)\in S]+\delta.
  \end{equation}
\end{definition}
\noindent $(\epsilon,\delta)$-DP ensures that for all adjacent datasets $\mathcal{D}$ and $\mathcal{D}^\prime$, the absolute value of the privacy loss will be bounded by $\epsilon$ with probability at least $1-\delta$, and the privacy loss is defined as:
\begin{equation}
    \mathcal{L}^x_{\mathcal{M}({\mathcal{D})}||\mathcal{M}({\mathcal{D}^\prime})}=\text{ln}(\frac{\text{P}[\mathcal{M}({\mathcal{D})=x]}}{\text{P}[\mathcal{M}(\mathcal{D}^\prime)=x]}).
    \label{eq:privacyloss}
\end{equation}
% The privacy budget $\epsilon$ controls the privacy protection level. A smaller $\epsilon$ value means a higher privacy level, which implies a lower possibility to distinguish the outputs from two adjacent datasets with the randomized mechanism $\mathcal{M}$. $\delta$ is the privacy violation probability.
If $\delta=0$, the randomized algorithm $\mathcal{M}$ is said to have $\epsilon$-DP.
% From~\cite{dwork2014algorithmic}, $(\epsilon,\delta)$-differential
% privacy ensures that for all adjacent datasets $\mathcal{D}$ and $\mathcal{D}^\prime$, the absolute value of the privacy loss will be bounded by $\epsilon$ with probability at least $1-\delta$, and the privacy loss is defined as
% \begin{equation}
%     \mathcal{L}^x_{\mathcal{M}({\mathcal{D})}||\mathcal{M}({\mathcal{D}^\prime})}=\text{ln}(\frac{\text{Pr}[\mathcal{M}({\mathcal{D})=x]}}{\text{Pr}[\mathcal{M}(\mathcal{D}^\prime)=x]}).
%     \label{eq:privacyloss}
% \end{equation}
A generic method of achieving $(\epsilon,\delta)$-DP is the Gaussian mechanism that adds Gaussian noise from the distribution $\mathcal{N}(0,\sigma^2)$, calibrated to the sensitivity. The sensitivity captures the maximum difference in the output of the dataset caused by a single record in the worst case. We define sensitive more formally below and then we describe the Gaussian mechanism.
\begin{definition}[\bf $l_p$-sensitive]
    Let $f$ be a query mapping from the space of datasets to $\mathbb{R}^m$, Let $N$ be the
set of all possible pairs of neighboring datasets. For a fixed positive scalar $p$, the $l_p$-sensitivity of $f$ is defined by
\begin{equation}
    s(f;p)=\max_{D,D'\in N}\|f(D)-f(D')\|_p.
\end{equation}
\end{definition}

\noindent\textbf{Gaussian mechanism.} Given the output of the query $f(D)\in\mathbb{R}^m$, the Gaussian mechanism adds noise sampled from the Gaussian distribution to each of the $m$ dimensions in the output, with the variance of the noise calibrated to the $l_p$-sensitivity. The Gaussian mechanism cannot guarantee pure $\epsilon$-DP but can instead ensure approximate $(\epsilon, \delta)$-DP. The Gaussian mechanism is commonly used in machine learning~\cite{ponomareva2023dp}, e.g., it is the main mechanism behind DP-SGD, which we are going to explain in Section~\ref{subsec:dpsgd}.
\subsection{DP-SGD Algorithm}\label{subsec:dpsgd}
DP-SGD is a privacy-preserving optimization algorithm designed to train machine learning models while ensuring individual data point privacy. The algorithm achieves this by incorporating DP directly into the stochastic gradient descent optimization process. Mathematically, DP-SGD modifies the gradient computation step by adding carefully calibrated noise, ensuring that the impact of any single data point on the model parameters remains statistically indistinguishable. The updated gradient computation can be expressed as:
\begin{equation}
    \tilde{\nabla}\ell(w)=\frac{\nabla \ell(w)}{\max(1,\frac{||\nabla \ell(w)||_2}C)}+\mathcal{N}(0,\sigma^2C^2I),
\end{equation}
where $\nabla \ell(w)$ is the gradient of the loss with respect to the model parameters for the data point. $C$ is the clipping norm, which bounds the influence of each example on the gradient and $\mathcal{N}(0,\sigma^2C^2I)$ represents the added noise following a Gaussian distribution with mean 0 and covariance $\sigma^2C^2I$.
The privacy cost is typically quantified using the $(\epsilon, \delta)$-DP definition. DP-SGD achieves DP by bounding the ratio of probabilities of any two adjacent datasets.
Incorporating these modifications into the standard stochastic gradient descent update rule, the DP-SGD algorithm provides a robust and privacy-preserving approach to training machine learning models on sensitive datasets.

The current method for realizing DP in QML was proposed by Watkins \textit{et al.} in~\cite{watkins2023quantum}, they artificially add Gaussian noise to the gradient like DP-SGD. In the next section, we additionally consider shot noise and incoherent noise, and we prove that the gradient satisfies a Gaussian distribution in a binary classification task using a modified hinge loss function, which can potentially achieve $(\epsilon,\delta)$-DP.
\section{NOISY GRADIENTS OF QUANTUM CIRCUIT PARAMETERS}\label{Sec.3}
In this section, we consider the binary classification task with modified hinge loss function on classical dataset. We leave the case of arbitrary loss functions for multi-classification tasks in the future work. The binary classification task takes $\mathcal{D}=\{(\mathbf{x}_i,\mathbf{y}_i)\in\mathcal{X}\times\mathcal{Y}:i=1,2,\ldots,m\}$ as input training data which consists of $m$ data-label pairs, where the data space $\mathcal{X}=\mathbb{R}^d$ and the label set $\mathcal{Y}\in\{-1,1\}$. We would like to train a classifier $\mathbf{g}:\mathcal{X}\to\mathcal{Y}$. We measure the quality of our classifier on the training data via a loss function $\ell:\mathcal{Y}\times\mathcal{Y}\to\mathbb{R}$. In empirical risk minimization (ERM), we choose a classifier $\mathbf{g}$ that minimizes the empirical loss:
\begin{equation}
J(\mathbf{g},\mathcal{D})=\frac1m\sum_{i=1}^m\ell(\mathbf{g}(\mathbf{x}_i),\mathbf{y}_i),
\end{equation}
where $\mathbf{g}(\mathbf{x}_i)$ is the output of the QNN (i.e. the expectation of an observable) under classical input data $\mathbf{x}_i$ and is parameterized by $\boldsymbol{\theta} = \{\theta_1,\theta_2,\cdots$\}. We will show that partial derivatives of the loss function $\nabla_\theta \ell$ satisfy the Gaussian distribution under the effect of shot noise, incoherent noise and the PEC method. Therefore, it can potentially achieve DP in QML.
% according to the DP-SGD algorithm proposed by~\cite{abadi2016deep}.
%\subsection{Shot Noise and Incoherent Noise in Quantum Circuits}
\subsection{The Impact of Shot Noise and Incoherent Noise on the Output in Quantum Circuits}
Given a QNN, the output has an expectation value of ${\rm Tr}(\mathcal{A}\rho)$, where $\rho$ is the output state of the quantum circuit, ${\rm Tr}(\cdot)$ is the trace function of a matrix and $\mathcal{A}$ is the observable of interests, whose eigenvalues are $\lambda_1,\lambda_2,...,\lambda_{2^q}$. Here, $q$ is the number of qubits to be observed. Assuming that after running the quantum circuit $n$ times, we can get the observation that $\mathcal{O}^n=\{\mathbf{o}_1,\mathbf{o}_2,...,\mathbf{o}_n\}$, where $\mathbf{o}_1,\mathbf{o}_2,...,\mathbf{o}_n$ are independently and identically distributed random variables. The probability that $\mathbf{o}_n=\lambda_k$ is $p_k$, which depends on the output state $\rho$. Consequently, the expectation value obtained by running the quantum circuit can be expressed as a random variable $O=\frac{\sum_{i=1}^n o_i}{n}$, where $o_i$ denotes the value of the random variable $\mathbf{o}_i$. According to the Lindeberg-Lévy central limit theorem~\cite{billingsley2017probability}, we assume that $O$ follows Gaussian distribution. It implies that each time a quantum computation is performed, the output value $\mathbf{g}(\mathbf{x}_i)$ is a random variable that follows Gaussian distribution.

In addition to shot noise, there is also inevitable incoherent noise in quantum circuits, which affects the variance of the output $O$. The wide variety of incoherent noise, coupled with the different specific forms of quantum circuits, makes it difficult to study the effect of such noise on the output $O$. Therefore, we approximate the incoherent noise in quantum circuits using the global depolarizing noise. It has been proven that global depolarizing noise successfully captures the effects of realistic noise and can serve as a phenomenological model~\cite{qin2023error}. The global depolarizing noise can be expressed as:
\begin{equation}
    \mathcal{E}^{\otimes q}_\text{Dep}(\rho)=(1-p_t)\rho+p_t\frac{I^{\otimes q}}{2^q},
    \label{globalDepolarizing}
\end{equation}
where $p_t$ is the effective total error probability and can be obtained by solving a quadratic equation~\cite{vovrosh2021simple}:
\begin{equation}
    \mathrm{Tr}\big[\rho^2\big]=(1-p_{t})^2+\frac{p_{t}(1-p_{t})}{2^{n-1}}+\frac{p_{t}^2}{2^n}.
\end{equation}
The value of $\mathrm{Tr}\big[\rho^2\big]$ can be measured directly on the device. Because $p_{t}$ affects the output quantum state, it affects the probability $p_k$ that $\lambda_k$ occurs. Therefore, the variance of the noise output $\hat{O}$ can be expressed as a function of the global depolarizing noise error rate $p_t$, the noiseless output quantum state $\rho$ and the number of executions $n$ of the quantum circuit:
\begin{equation}
    \text{Var}[\hat{O}]=h(p_t,\rho, n).
    \label{minmax}
\end{equation}
Without dealing with incoherent noise, the quantum circuit output results will be highly biased, which can have serious implications for the performances of QML. It will be shown in Section~\ref{4-A}. To address the effect of incoherent noise on QML, in the next section we consider the use of the PEC method and analyze its effect on the Gaussian distributed output $\hat{O}$.

\subsection{Analyzing the Variance of the Gaussian Distributed Output with Probabilistic Error Cancellation}\label{3-B}
The main idea of PEC is to represent each ideal gate $\mathcal{G}_i$ of the quantum circuit as a linear combination of noisy implementable operations $\{\mathcal{U}_\alpha\}$:
\begin{equation}
\mathcal{G}_i=\sum\limits_{\alpha}\eta_{i,\alpha}\mathcal{U}_{i,\alpha}, 
\label{noiseOperation}
\end{equation}
\noindent where $\eta_{i,\alpha}\in\mathbb{R}$, and $\sum\limits_\alpha\eta_{i,\alpha}=1$. The real coefficients $\eta_{i,\alpha}$ form a quasi-probability distribution. We can obtain the ideal expectation value as a linear combination of noisy expectation values:
\begin{equation}
\begin{split}
    \langle\mathcal{A}\rangle_\text{ideal} & = \sum_{\Vec{\alpha}}\eta_{\Vec{\alpha}}\text{Tr}[\mathcal{A}\Phi_{\Vec{\alpha}}(\rho)]\\
    & = \sum_{\Vec{\alpha}}\eta_{\Vec{\alpha}}\langle\mathcal{A}_{\Vec{\alpha}}\rangle_\text{noisy},
    \label{eq:pecIdeal}
\end{split}
\end{equation}
\noindent where $\Vec{\alpha}=(\alpha_1,\alpha_2,...,\alpha_t)$ is the multi-index, $\Phi_{\Vec{\alpha}}=\mathcal{U}_{t,\alpha_t}\circ\cdots\circ\mathcal{U}_{2,\alpha_2}\circ\mathcal{U}_{1,\alpha_1}$, $\eta_{\Vec{\alpha}}=\prod_{i=1}^t\eta_{i,\alpha_i}$. 
The noise expectation values $\langle\mathcal{A}_{\Vec{\alpha}}\rangle_\text{noisy}$ can be measured directly. By combining all the noise expectation values, the desired result $\langle\mathcal{A}\rangle_\text{ideal}$ can be calculated.

Inevitably, we need to investigate how PEC affects the Gaussian distribution output $O$ of the quantum circuit. We denote $\Tilde{O}$ as the output random variable after using PEC. From Eq.~\eqref{eq:pecIdeal}, the error mitigated expectation value $\mathbb{E}[{\Tilde{O}}]$ is a linear sum of the
expectation values of a set of $K$ random variables $\{X_i\}$, which are the outputs of the set of measurement circuits, and $K$ corresponds to the multi-index $\Vec{\alpha}$. Therefore, $\mathbb{E}[\Tilde{O}]$ can be expressed as:
\begin{equation}
    \mathbb{E}[\Tilde{O}]=\sum\limits_{i=1}^K\eta_i\mathbb{E}[X_i]=\mathbb{E}[O],
    \label{unbiased}
\end{equation}
\noindent where $\eta_i$ are real coefficients. If we estimate individual terms $\mathbb{E}[X_i]$ up to a certain precision, then combine the results, the variance of the error mitigation estimator $\Tilde{O}$ can be expressed as:
\begin{equation}
    \text{Var}[\Tilde{O}]=\sum\limits_{i=1}^K|\eta_i|^2\text{Var}[X_i].
\end{equation}
Because the component of random variables $X_i$ are generated from circuits that are variants of the primary circuit, they can be expected to have a similar variance as the unmitigated estimator $\hat{O}$~\cite{cai2022quantum}. This relationship can be expressed as $\text{Var}[X_i]\sim\text{Var}[\hat{O}]$, and hence we have:
\begin{equation}
    \text{Var}[\Tilde{O}]=(\sum\limits_{i=1}^K|\eta_i|^2)\text{Var}[\hat{O}].
    \label{varianceAfterQEM}
\end{equation}

So far we have given analytical expressions for the mean and variance of a error mitigated quantum circuit's output $\Tilde{O}$. In the next section, we will analyze how much privacy budget can be given to a quantum circuit.

\subsection{Quantum Inherent Noise Mechanism}\label{3-C}
In this subsection, we propose a novel DP mechanism called the quantum inherent noise mechanism, which utilizes the inherent noise in the quantum circuit to achieve $(\epsilon,\delta)$-DP.
The partial derivative of the loss function $\nabla_\theta \ell$ can often be expressed as a linear combination of two quantum circuit outputs. The two quantum circuits use the same architecture, differing only in $\pm\frac\pi 2$ shift of the parameter. Here, we use parameter-shift rules~\cite{mitarai2018quantum} to compute partial derivatives of observable $\mathcal{A}$ with respect to the parameter $\theta$:
% Here, we use parameter-shift rules~\cite{mitarai2018quantum,schuld2019evaluating} to compute partial derivatives of parameter $\theta$:
\begin{equation}
    \frac{\partial\langle \mathcal{A}\rangle}{\partial\theta}=\frac12\left[\mathbf{g}(\mathbf{x}_i;\theta+\frac\pi2))-\mathbf{g}(\mathbf{x}_i;\theta-\frac\pi2))\right].
    \label{parameter-shift rules}  
\end{equation}
If we consider using modified hinge loss function $\ell(\mathbf{g}(\mathbf{x}_i),\mathbf{y}_i)=1-\mathbf{y}_i\mathbf{g}(\mathbf{x}_i)$, the gradient of the parameter $\theta$ with respect to the data point $(\mathbf{x}_i,\mathbf{y}_i)$ can be expressed as:
\begin{equation}
    \nabla_\theta \ell=\pm\frac12\left[\mathbf{g}(\mathbf{x}_i;\theta+\frac\pi2)-\mathbf{g}(\mathbf{x}_i;\theta-\frac\pi2)\right].
    \label{hingloss}
\end{equation}
% \begin{equation}
% \begin{split}
%     &\nabla_\theta l(\mathbf{g}(\mathbf{x}_i),\mathbf{y}_i;\theta)\\
%     &=\frac12\left[\mathbf{y}_i\hat{\mathbf{y}}_i^{\theta-\frac\pi2}-\mathbf{y}_i\hat{\mathbf{y}}_i^{\theta+\frac\pi2}\right]\\
%     &=\frac12\left[\pm\mathbf{g}(\mathbf{x}_i;\theta-\frac\pi2)\pm\mathbf{g}(\mathbf{x}_i;\theta+\frac\pi2)\right].
%     \label{final}
% \end{split}
% \end{equation}
We mentioned earlier that $\mathbf{g}(\mathbf{x}_i;\theta)$ is mitigated quantum circuits output satisfying Gaussian distribution. Since $\mathbf{g}(\mathbf{x}_i;\theta+\frac\pi2)$ and $\mathbf{g}(\mathbf{x}_i;\theta-\frac\pi2)$ are independent Gaussian random variables, their linear combination $\nabla_\theta \ell$ also satisfies Gaussian distribution, and we have the following Lemma.
\begin{lemma}
    %The variance of the Gaussian distribution of quantum gradients $\text{Var}[\nabla_\theta l]$ is bounded by:
    The gradient of the loss function in QML model satisfies a Gaussian distribution with the variance $\text{Var}[\nabla_\theta \ell]$ that is bounded by:
    \begin{equation}
       \frac{\sum\limits_{i=1}^K|\eta_i|^2}{2}h_\text{min} \leq \text{Var}[\nabla_\theta \ell] \leq \frac{\sum\limits_{i=1}^K|\eta_i|^2}{2}h_\text{max},
       \label{upperandlower}
    \end{equation}
    where $h_\text{min}=\min\limits_\rho h(p_t,\rho,n)$ and $h_\text{max}=\max\limits_\rho h(p_t,\rho,n)$.
    \label{lemma:1}
\end{lemma}
\begin{proof}
Using the fact $\Tilde{O}^\theta=\mathbf{g}(\mathbf{x}_i;\theta)$ and combining Eq.~(\ref{minmax})(\ref{varianceAfterQEM})(\ref{hingloss}), we can deduce that:
\begin{equation}
\begin{split}
    &\text{Var}[\nabla_\theta \ell]\\
    &=\frac14\text{Var}[\mathbf{g}(\mathbf{x}_i;\theta+\frac\pi2)]+\frac14\text{Var}[\mathbf{g}(\mathbf{x}_i;\theta-\frac\pi2)]\\
    &=\frac14\text{Var}[ \Tilde{O}^{\theta+\frac\pi2}]+\frac14\text{Var}[\Tilde{O}^{\theta-\frac\pi2}]\\
    &=\frac{\sum\limits_{i=1}^K|\eta_i|^2}{4}\text{Var}[\hat{O}^{\theta+\frac\pi2}]+\frac{\sum\limits_{i=1}^K|\eta_i|^2}{4}\text{Var}[\hat{O}^{\theta-\frac\pi2}]\\
    &=\frac{\sum\limits_{i=1}^K|\eta_i|^2}{4}h(p_t,\rho^{\theta+\frac\pi2},n)+\frac{\sum\limits_{i=1}^K|\eta_i|^2}{4}h(p_t,\rho^{\theta-\frac\pi2},n).
    \nonumber
    \label{var min}
\end{split}
\end{equation}
Since the function $h$ represents the noise output's variance, it is a bounded function. We can obtain the lower bounds $h_\text{min}$ and upper bounds $h_\text{max}$ on $h$. Bringing the minimum and maximum values into Eq.~\eqref{var min}, we can obtain lower and upper bounds on $\text{Var}[\nabla_\theta \ell]$.
\end{proof}
From Lemma \ref{lemma:1}, we can bound the variance of the Gaussian distributed gradients $\text{Var}[\nabla_\theta \ell]$. Combining an auxiliary result in~\cite{xu2020adaptive}, we can use the lower bound on $\text{Var}[\nabla_\theta \ell]$ to compute privacy budget $\epsilon$ and the protection failure probability $\delta$ of DP, which will be shown in Theorem~\ref{theorem:1}. In our proposed quantum inherent noise mechanism, we harness the shot noise and incoherent noise as a free noise resource in the quantum circuit to achieve $(\epsilon,\delta)$-DP instead of injecting extra artificial noise into the quantum circuits.
%The result is omitted for space, and we have the following Theorem.
% Consequently, we can use the following theorem to compute $\epsilon$ and $\delta$ of our proposed quantum inherent noise mechanism.

\begin{theorem}
    Suppose mechanism $M(D)=f(D)+Z$, where $D$ is the input dataset, $f(\cdot)$ is a $m$-dimension function, noise $Z\sim \mathcal{N}(0,\sigma_*^2)$ and $\sigma_*=\sigma_\text{min}$, is $(\epsilon,\delta)$-DP, then quantum inherent noise mechanism $M^\prime (D)=f(D)+Z^\prime$, where $Z^\prime=(z_{1},\ldots,z_{m})^{T}$, $\forall i\in[m],z_i\sim \mathcal{N}(0,\sigma_i^2)$ and $\sigma_i\in[\sigma_\text{min},\sigma_\text{max}]$ is also $(\epsilon,\delta)$-DP.
    \label{theorem:1}
\end{theorem}
% \noindent To prove Theorem \ref{theorem:1}, we need an auxiliary result in~\cite{xu2020adaptive} and it is omitted for space. 
\begin{proof}
    Using Lemma 1 in~\cite{xu2020adaptive}, We can introduce an equivalent $\sigma_*$ which satisfies
    \begin{equation}    {\sigma_*^{2}}\leq\min_i\sigma_{i}^{2}=\frac{(\Delta_2f)^2}{\frac{(\Delta_2f)^2}{\min_i\sigma_{i}^{2}}}=\frac{(\Delta_2f)^2}{\frac{\sum_{i=1}^{m}s_{i}^{2}}{\min_i\sigma_{i}^{2}}}\leq\frac{(\Delta_2f)^2}{\sum_{i=1}^{m}\frac{s_{i}^{2}}{\sigma_{i}^{2}}},
    \nonumber
    \end{equation}
    where $\sigma^2_i$ is the variance of the $i$-th gradient whose lower bound is $\sum\limits_{i=1}^K|\eta_i|^2h_\text{min}/2$ from Lemma~\ref{lemma:1} and $s_i$ is the $l_2$ sensitivity of the $i$-th gradient. Consequently, we can use the upper bound on $\sigma_*^2$ and Eq.~\eqref{eq:privacyloss} to compute optimal $\epsilon$ and $\delta$.
    %The proof is completed.
\end{proof}

\begin{figure*}[h]
    \centering
  %   \subfigure[Feature map of the quantum circuit.]{
		% \includegraphics[width=2.2in]{figs/Fig0.pdf}
  %       \label{fig0a}}
  %   \subfigure[Ansatz of the quantum circuit.]{
		% \includegraphics[width=2.2in]{figs/Fig-1.pdf}
  %       \label{fig0b}}
    \includegraphics[width=6in]{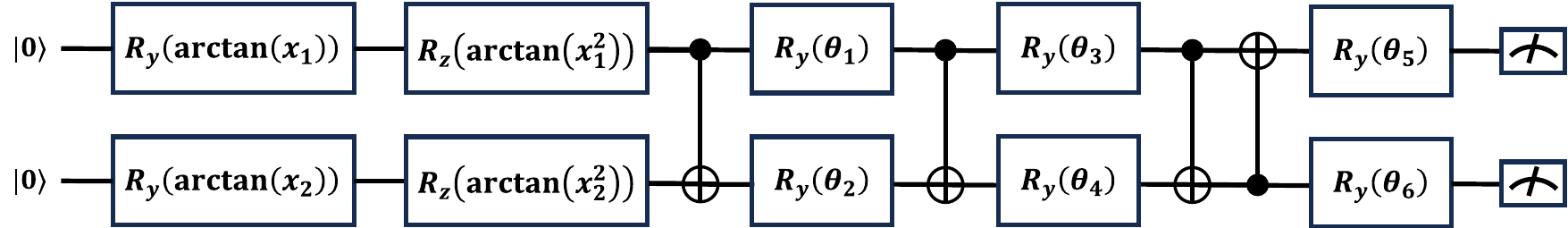}
    \caption{Quantum circuit for 2D classification of Iris dataset.}
    \label{fig0}
\end{figure*}

From Theorem \ref{theorem:1}, we can determine that each step of parameter update for QML is ($\epsilon$,$\delta$)-DP. Using the composition theorem of DP, we can compute the final $\epsilon$ and $\delta$ like~\cite{abadi2016deep} and the value of the gradient clipping parameter $C$ can take the maximum value of the $l_2$ norm of gradients by multiple pre-training.

\section{EVALUATION}\label{EVALUATION}
We conduct simulation experiments in this section to illustrate the impact of inherent noise on the privacy budget of QML across different error rates and shots.
We utilize Iris dataset to generate the binary classification dataset. We choose the Petal Length and Petal Width features of the Iris Setosa and Iris Versicolour categories with corresponding labels of $-1$ and $1$.
The QNN is constructed and trained utilizing open-source software packages \textit{Qiskit} and we implement PEC by \textit{Mitiq}. The structure of the QNN we use is shown in Fig.~\ref{fig0}. 
For better numerical stability and convergence, we squeeze all features onto the interval $[0,1]$. The observable we use is the tensor product of pauli matrix $I\otimes Z$. What's more, we use a modified hinge loss function and the proportion of training and testing is 80\% and 20\%. Optimization is performed using the gradient descent method with a learning rate set to 0.01. In the simulation experiments, if not specified, we assumed an error probability of 5\% for single qubit gates, 10\% for two qubit controlled gates.

% We use two features $(x_1,x_2)$ of Iris dataset for the binary classification task and we implement the QNN model with Qiskit and Mitiq. We normalize the features which is a common technique in machine learning and often leads to better numerical stability and convergence of an algorithm. In our case, we squeeze all features onto the interval $[0,1]$. Our choice of feature map follows~\cite{watkins2023quantum}. In the simulation experiments we assumed an error probability of 5\% for single qubit quantum gates, 10\% for two qubit controlled gates. The observable we use is the tensor product of pauli matrix $I\otimes Z$. We use gradient descent method and set the learning rate to 0.01. The structure of the quantum neural network we use is shown in Fig.~\ref{fig0}.
\subsection{Impact of Incoherent Noise on QML}\label{4-A}
In this section we investigate how incoherent noise affects QML's test accuracy.
As shown in Fig.~\ref{fig1}, when there is unprocessed incoherent noise, QML's performance can be severely impacted, which is manifested in two-fold: (1) incoherent noise affects the speed of the model convergence, and (2) incoherent noise may cause the model fail to converge. Incoherent noise can lead to errors in the quantum computation, especially deep quantum circuits, which can propagate through the quantum circuit, affecting subsequent operations and measurements. This propagation can be particularly detrimental in complex quantum algorithms used in machine learning tasks. Therefore it is necessary to deal with incoherent noise in QML scenarios.
\begin{figure}[ht]
    \centering
    \includegraphics[width=3.2in]{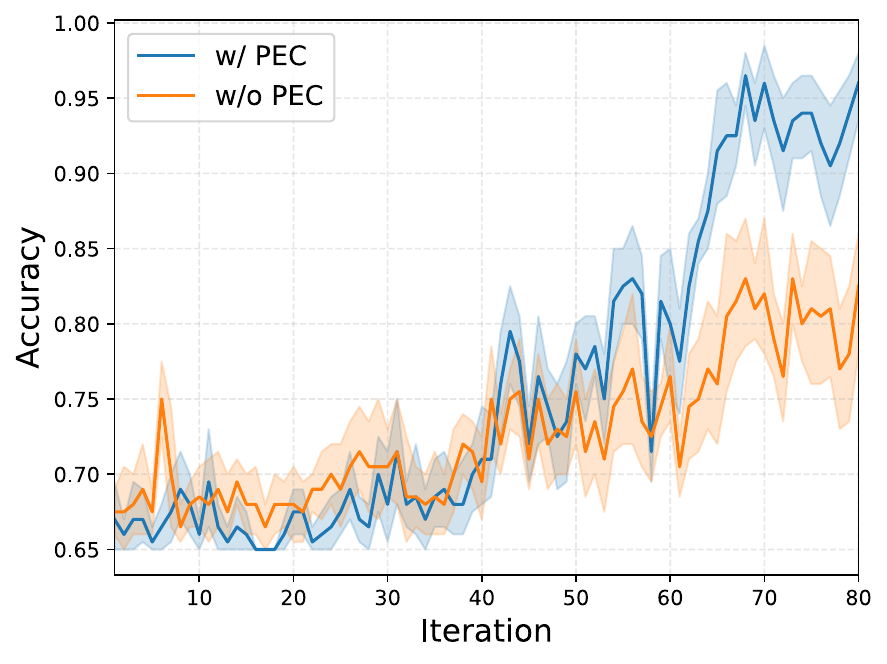}
    \caption{Comparison of test accuracy with and without the use of PEC.}
    \label{fig1}
\end{figure}
% \subsection{Effect of Different $c$ on QML}
% In Sec.\ref{3-C} we limited the output range of the quantum circuit by $c$ in order to satisfy condition Eq.(\ref{Condition}), here we study the effect of different $c$ on quantum machine learning. We draw the images of the loss function and the test accuracy function for different $c$ through simulation, as shown in Fig.\ref{fig2}.
% \begin{figure}[htbp]
% \centering
%     \subfigure[xx]{
% 		\includesvg[width=2.2in]{figs/Fig2a.svg}\label{fig2a}}
%     \subfigure[xx]{
% 		\includesvg[width=2.2in]{figs/Fig2b.svg}\label{fig2b}}
%     \caption{xxx.}
%     \label{fig2}
% \end{figure}
% Through Fig.\ref{fig2a} we are able to see that the larger the $c$ used, the faster the model converges for the same learning rate, but the larger the $c$ leads to larger fluctuations in the test accuracy of the final converged model. With Figure Fig.\ref{fig2b} we are able to see that at the beginning of training, the larger the c the greater the training loss, and after the model converges, the larger the $c$ the smaller the training loss. The reason for these phenomena in the simulation is that $c$ affects Eq.(\ref{MAE Loss}), which in turn affects Eq.(\ref{parameter-shift rules}). Smaller $c$ leads to smaller losses and also smaller gradients, so smaller $c$ leads to slower model convergence. However, the smaller $c$ computes a smaller gradient when the model finally converges, and therefore the test accuracy fluctuates less.
\subsection{Variance of the Gaussian Distribution Gradient}
According to Theorem~\ref{theorem:1}, we conduct experiment to elucidate how much variance can bring to the QNN's parameter gradients in this section. First, we fix the incoherent noise to calculate the lower bound on the standard deviation that different shots can bring to the gradient. As shown in Fig.~\ref{fig2}, the standard deviation of the gradient decreases as quantum circuits are executed more number of times, which is determined by Eq.~\eqref{upperandlower} in Lemma \ref{lemma:1}.
Next, when the quantum circuits are executed the same number of times, we consider the effect of the global depolarizing noise probability on the standard deviation's lower bound. As shown in Fig.~\ref{fig3}, when the noise probability increases, the lower bound on the standard deviation also increases, and therefore the gradient carries more noise and achieves better privacy preservation, which will be analyzed quantitatively in Section \ref{4-C}.
\begin{figure}[htbp]
    \centering
    \subfigure[]{
		\includegraphics[width=1.6in]{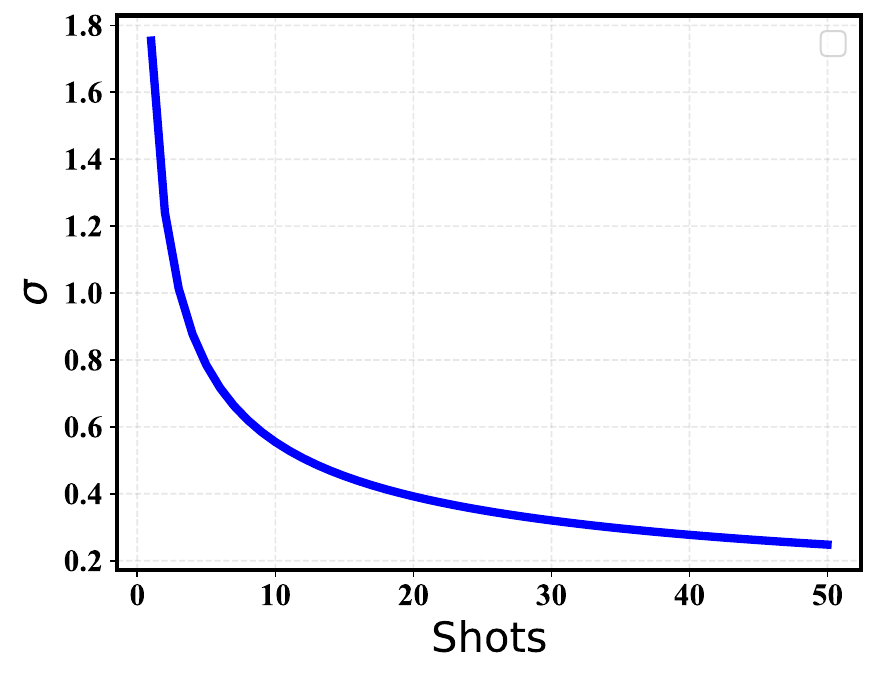}
        \label{fig2}}
    \hfill
    \centering
    \subfigure[]{
		\includegraphics[width=1.6in]{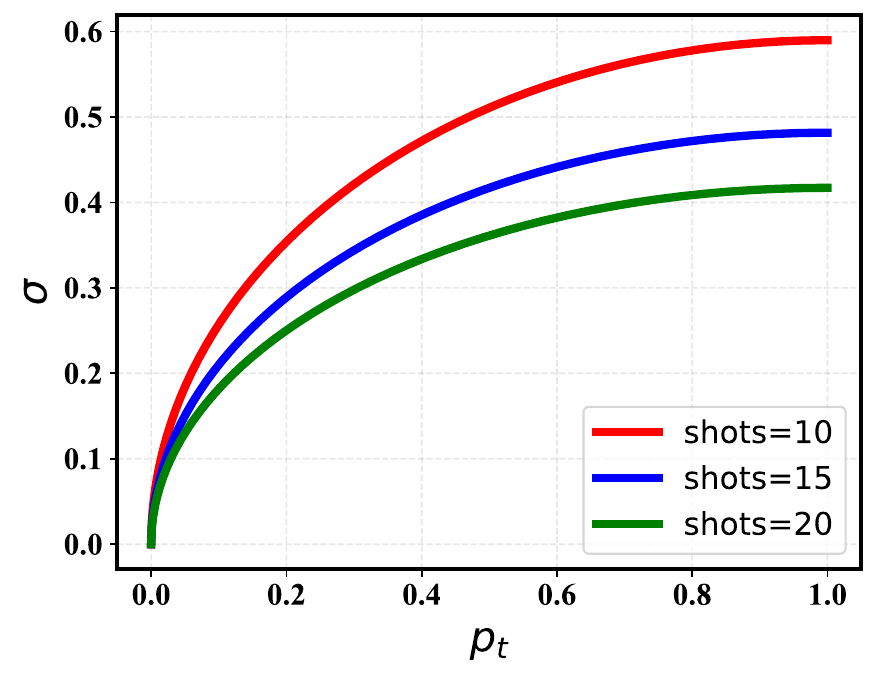}
        \label{fig3}}
        \caption{(a) The impact of shots on the standard deviation's lower bound. (b) The impact of global depolarizing noise probability on the standard deviation's lower bound.}
\end{figure}
% \begin{figure}[htbp]
%     \centering
%     \includesvg[width=3in]{figs/Fig3.svg}
%     \caption{The effect of the global depolarizing noise probability $p_t$ on the lower bound of the standard deviation of the Gaussian distribution gradient for different number of quantum circuit executions.}
%     \label{fig3}
% \end{figure}
\begin{figure}[h]
    \centering
    \subfigure[]{
		\includegraphics[width=1.6in]{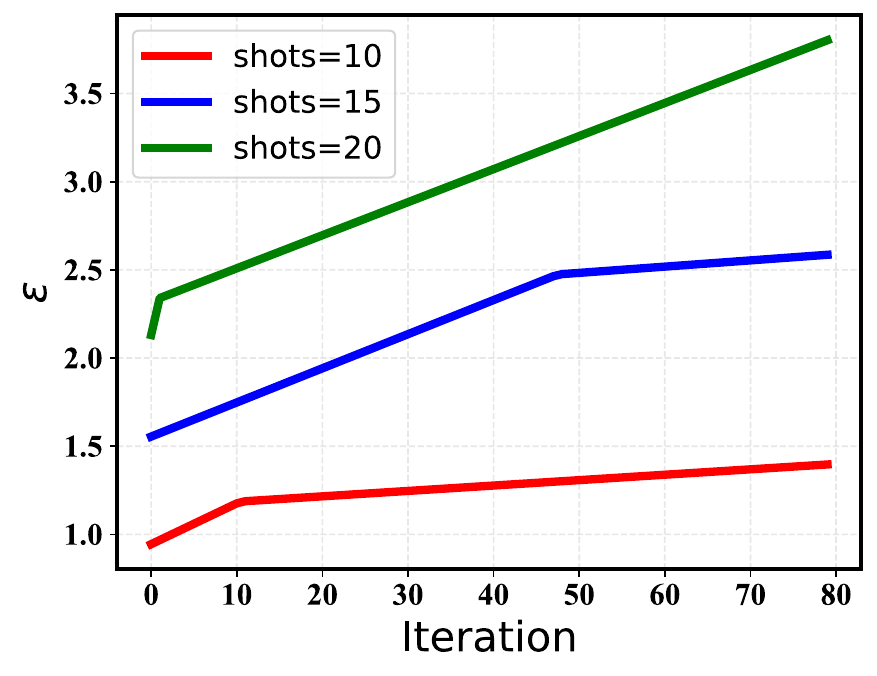}
        \label{fig4}}
    \hfill
    \centering
    \subfigure[]{
		\includegraphics[width=1.6in]{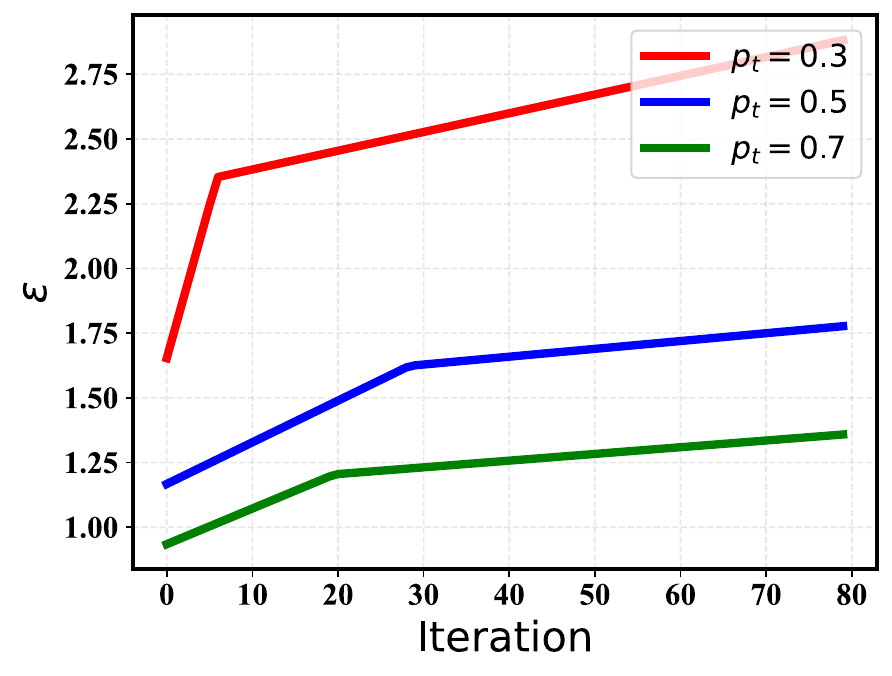}
        \label{fig5}}
        \caption{(a) The value of $\epsilon$ in different shots and iterations. (b) The value of $\epsilon$ in different global depolarizing error rates and iterations.}
\end{figure}

\subsection{Privacy Preserving QML}\label{4-C}
In this section, We analyze how much privacy budget $\epsilon$ can be achieved to a QML model by different shots and global depolarizing noise. In this experiments the default value of the privacy violation probability $\delta$ is set to 0.01 and the default value of the gradient clipping parameter $C$ is 0.7 based on experiments. We first analyze the relationship between shots and the privacy budget $\epsilon$. As shown in Fig.~\ref{fig4}, with iteration determined, the privacy budget $\epsilon$ increases when shots increase, this is because larger shots lead to a decrease of the gradients' standard deviation, and therefore an increase in privacy budget.
% \begin{figure}[htbp]
%     \centering
%     \includesvg[width=3in]{figs/Fig4.svg}
%     \caption{The $\epsilon$ value as a function of iteration for different shots and $\delta=0.01$.}
%     \label{fig4}
% \end{figure}
Next, we analyze the effect of the global depolarizing noise probability $p_t$ on the privacy budget $\epsilon$ under fixed 10 shots, which is shown in Fig.~\ref{fig5}. As the probability of noise increases, the privacy budget $\epsilon$ decreases. This is because larger noise probability leads to an increase in the lower bound on the standard deviation, which can better protect the model privacy.
% \begin{figure}[htbp]
%     \centering
%     \includesvg[width=3in]{figs/Fig5.svg}
%     \caption{The $\epsilon$ value as a function of iteration for different global depolarizing error rate $p_t$ and $\delta=0.01$.}
%     \label{fig5}
% \end{figure}
\section{CONCLUTION AND FUTURE WORK}\label{Sec.5}
In this paper, we first introduce DP and the DP-SGD algorithm for preserving privacy and show that shot noise and incoherent noise can potentially achieve DP in QML. Next, we show that the output of a quantum circuit follows a Gaussian distribution under the influence of shot noise, incoherent noise, and the PEC method. What's more, we give an analytic expression for the variance of the Gaussian distribution. Finally, we prove that the gradients of the quantum circuit parameters also follow a Gaussian distribution when using a modified hinge loss function for a binary classification task, which can protect QML models' privacy.

However, we only considered PEC as a QEM method. In future work, we propose to study the achievable DP budget under different QEM methods. Moreover, we only studied QML under modified hinge loss function for the binary classification task. We will investigate how to generalize it to QML under arbitrary loss functions for multi-classification tasks. Furthermore, we used a lower bound on the variance in the privacy budget calculation, but the actual variance may depend on the specific task and quantum circuits. Given the QML task, it is worth studying how to precisely calculate the privacy budget and provide a tighter privacy bound.

\section*{Acknowledgments}
This work was supported by NSFC Project 61801045.

\bibliographystyle{IEEEtran}
\bibliography{./ref}

\end{document}